\documentclass[11pt]{article}

\usepackage[paper=letterpaper, margin=1in]{geometry}

\usepackage{ifthen}

\ifthenelse{\isundefined{\GenerateShortVersion}}
{
\def\LongVersion{}
\def\LongVersionEnd{}
\long\def\ShortVersion#1\ShortVersionEnd{}
}
{
\def\ShortVersion{}
\def\ShortVersionEnd{}
\long\def\LongVersion#1\LongVersionEnd{}
}

\newcommand{\Ignore}[1]{\ignorespaces}

\usepackage{amsmath, amssymb}

\usepackage{amsthm}

\usepackage{ifpdf}

\usepackage{authblk}

\ifpdf
\usepackage[pdftex]{graphicx}
\else
\usepackage[dvips]{graphicx}
\fi

\usepackage[%
  ocgcolorlinks,%
  linkcolor=blue,%
  filecolor=blue,%
  citecolor=blue,%
  urlcolor=blue]{hyperref}

\usepackage{enumitem}

\usepackage{xcolor}

\usepackage{mdframed}

\LongVersion 
\LongVersionEnd 

\ShortVersion 
\usepackage{times}
\ShortVersionEnd 

\ShortVersion 
\renewcommand{\paragraph}[1]{\par\noindent\textbf{#1}}
\ShortVersionEnd 

\newtheorem{theorem}{Theorem}[section]
\newtheorem{lemma}[theorem]{Lemma}
\newtheorem{observation}[theorem]{Observation}
\newtheorem{corollary}[theorem]{Corollary}

\theoremstyle{definition}

\theoremstyle{plain}

\newenvironment{DenseItemize}[0]
{\begin{itemize}[nosep, leftmargin=*]}
{\end{itemize}}

\LongVersion 
\newenvironment{MathMaybe}[0]
{\begin{displaymath}\ignorespaces}
{\end{displaymath}}
\LongVersionEnd 
\ShortVersion 

\ShortVersionEnd 

\newenvironment{IntuitionSpotlight}[0]
{\begin{mdframed}[%
backgroundcolor=gray!30,%
topline=false,%
bottomline=false,%
linewidth=1pt]%
\noindent\textbf{Intuition spotlight:}}
{\end{mdframed}}

\newcommand{\Integers}{\mathbb{Z}}
\newcommand{\GrowBall}{\mathtt{GrowBall}}
\newcommand{\Broadcast}{\mathtt{Broadcast}}
\newcommand{\Echo}{\mathtt{Echo}}
\newcommand{\Proceed}{\mathtt{Proceed}}
\newcommand{\RandSymbol}{\mathtt{RandSymbol}}
\newcommand{\SAloops}{\ensuremath{\mathrm{SA}^{\circlearrowleft}}}

\newcommand{\Sect}{Sec.}
\newcommand{\Thm}{Thm.}
\newcommand{\Lem}{Lem.}
\newcommand{\Obs}{Obs.}

\newcommand{\Fig}{Fig.}

\title{Selecting a Leader in a Network of Finite State Machines}

\author[1]{Yehuda Afek\thanks{%
The work of Y.~Afek was partially supported by a grant from the Blavatnik
Cyber Security Council and the Blavatnik Computer Science Research Fund.}}
\author[2]{Yuval Emek\thanks{%
The work of Y.~Emek was supported in part by an Israeli Science
Foundation grant number 1016/17.}}
\author[3]{Noa Kolikant}
\affil[1]{Tel Aviv University.
\texttt{afek@cs.tau.ac.il}}
\affil[2]{Technion.
\texttt{yemek@technion.ac.il}}
\affil[3]{Tel Aviv University.
\texttt{noakolikant@mail.tau.ac.il}}

\date{}

\begin{document}

\begin{titlepage}

\maketitle

\begin{abstract}
This paper studies a variant of the \emph{leader election} problem under the
\emph{stone age} model (Emek and Wattenhofer, PODC 2013) that considers a
network of $n$ randomized finite automata with very weak communication
capabilities (a multi-frequency asynchronous generalization of the
\emph{beeping} model's communication scheme).
Since solving the classic leader election problem is impossible even in more
powerful models, we consider a relaxed variant, referred to as
\emph{$k$-leader selection}, in which a leader should be selected out of at
most $k$ initial candidates.
Our main contribution is an algorithm that solves $k$-leader selection for
bounded $k$ in the aforementioned stone age model.
On (general topology) graphs of diameter $D$, this algorithm runs in
$\tilde{O}(D)$ time and succeeds with high probability.
The assumption that $k$ is bounded turns out to be unavoidable:
we prove that if
$k = \omega (1)$,
then no algorithm in this model can solve $k$-leader selection with a
(positive) constant probability.
\end{abstract}

\textbf{keywords}:
stone age model,
beeping communication scheme,
leader election,
$k$-leader selection,
randomized finite state machines,
asynchronous scheduler

\end{titlepage}

\section{Introduction}
\label{sec:introduction}
Many distributed systems rely on the existence of one distinguishable node,
often referred to as a \emph{leader}.
Indeed, the \emph{leader election} problem is among the most extensively
studied problems in distributed computing
\cite{GallagerHS1983, Awerbuch1987, LavelleeL1990, AfekM1994}.
Leader election is not confined to digital computer systems though as the
dependency on a unique distinguishable node is omnipresent in \emph{biological
systems} as well
\cite{KellerN1993, SetchellCW2005, KuzdzalFickQS2010}.
A similar type of dependency exists also in networks of man-made micro- and
even nano-scale sub-microprocessor devices \cite{DerakhshandehGSRRS2015}.

The current paper investigates the task of electing a leader in networks
operating under the \emph{stone age (SA)} model \cite{EmekW2013} that provides
an abstraction for distributed computing by nodes that are significantly
inferior to modern computers in their computation and communication
capabilities.
In this model, the nodes are controlled by randomized finite automata and can
communicate with their network neighbors using a fixed message alphabet based
on a weak communication scheme that can be viewed as an asynchronous
extension of the \emph{set broadcast (SB)} communication model of
\cite{HellaJKLLLSV2015}
(a formal definition of our model is provided in \Sect{}~\ref{sec:model}).

Since the state space of a node in the SA model is fixed and does not grow
with the size of the network, SA algorithms are inherently \emph{uniform},
namely, the nodes are anonymous and lack any knowledge of the network size.
Unfortunately, classic impossibility results state that leader election is
hopeless in these circumstances (even under stronger computational models):
Angluin \cite{Angluin1980} proved that uniform algorithms cannot solve leader
election in a network with success probability $1$;
Itai and Rodeh \cite{ItaiR1990} extended this result to algorithms that are
allowed to fail with a bounded probability.

Thus, in the distributed systems that interest us, leader election cannot be
solved by the nodes themselves and some ``external help'' is necessary.
This can be thought of as an external \emph{symmetry breaking signal} that
only one node is supposed to receive.
Symmetry breaking signals are actually quite common in reality and can come in
different shape and form.
A prominent example for such external signaling occurs during the development
process of multicellular organisms, when ligand molecules flow through a
cellular network in a certain direction, hitting one cell before the others
and triggering its differentiation \cite{Slack2009}.

But what if the symmetry breaking signal is \emph{noisy} and might be received
by a handful of nodes?
Is it possible to detect that several nodes received this signal?
Can the system recover from such an event or is it doomed to operate with
multiple leaders instead of one?

In this paper, we study the \emph{$k$-leader selection} problem, where at most
$k$ (and at least $1$) nodes are initially marked as \emph{candidates}, out of
which exactly one should be selected.
On top of the relevance of this problem to the aforementioned questions, it is
also motivated by the following application.
Consider scenarios where certain nodes, including the leader, may get lost
during the network deployment process, e.g., a sensor network whose nodes are
dropped from an airplane.
In such scenarios, one may wish to produce
$k > 1$
candidate leaders with the purpose of increasing the probability that at least
one of them survives;
a $k$-leader selection algorithm should then be invoked to ensure that the
network has exactly one leader when it becomes operational.

The rest of the paper is organized as follows.
In \Sect{}~\ref{sec:model}, we provide a formal definition of the distributed
computing model used in the paper.
Our results are summarized in \Sect{}~\ref{sec:results} and some additional
related literature is discussed in \Sect{}~\ref{sec:related-literature}.
A $k$-leader selection algorithm that constitutes our main technical
contribution, is presented in \Sect{}~\ref{sec:algorithm}, whereas
\Sect{}~\ref{sec:negative-results} provides some negative results.

\subsection{Model}
\label{sec:model}
The distributed computing model considered in this paper follows the
\emph{stone age (SA)} model of Emek and Wattenhofer \cite{EmekW2013}.
Under this model, the communication network is represented by a finite
connected undirected graph
$G = (V, E)$
whose nodes are controlled by \emph{randomized finite automata} with state
space $Q$, message alphabet $\Sigma$, and transition function $\tau$ whose
role is explained soon.

Each node
$v \in V$
of degree $d_{v}$ is associated with $d_{v}$ \emph{input ports} (or simply
\emph{ports}), one port $\psi_{v}(u)$ for each neighbor $u$ of $v$ in $G$,
holding the last message
$\sigma \in \Sigma$
received from $u$ at $v$.
The communication model is defined so that when node $u$ sends a message, the
same message is delivered to all its neighbors $v$;
when (a copy of) this message reaches $v$, it is written into port
$\psi_{v}(u)$, overwriting the previous message in this port.
Node $v$'s (read-only) access to its own ports $\psi_{v}(\cdot)$ is very
limited:
for each message type
$\sigma \in \Sigma$,
it can only distinguish between the case where $\sigma$ is not written in any
port $\psi_{v}(\cdot)$ and the case where it is written in at least one port.

The execution is event driven with an asynchronous scheduler that schedules
the aforementioned message delivery events as well as node activation events.%
\footnote{The only assumption we make on the event scheduling is FIFO message
delivery:
a message sent by node $u$ at time $t$ is written into port $\psi_{v}(u)$ of
its neighbor $v$ before the message sent by $u$ at time
$t' > t$.}
When node
$v \in V$
is activated, the transition function
$\tau : Q \times \{0,1\}^{\Sigma} \rightarrow 2^{Q \times \Sigma}$
determines (in a probabilistic fashion) its next state
$q' \in Q$
and the next message
$\sigma' \in \Sigma$
to be sent based on its current state
$q \in Q$
and the current content of its ports.
Formally, the pair
$(q', \sigma')$
is chosen uniformly at random from
$\tau(q, \chi_{v})$,
where 
$\chi_{v} \in \{ 0, 1 \}^{\Sigma}$
is defined so that
$\chi_{v}(\sigma) = 1$
if and only if $\sigma$ is written in at least one port $\psi_{v}(\cdot)$.

To complete the definition of the randomized finite automata, one has to
specify
the set
$Q_{in} \subseteq Q$
of initial states that encode the node's input,
the set
$Q_{out} \subseteq Q$
of output states that encode the node's output,
and
the initial message
$\sigma_{0} \in \Sigma$
written in the ports when the execution begins.
SA algorithms are required to have \emph{termination detection}, namely, every
node must eventually decide on its output and this decision is irrevocable.

Following the convention in message passing distributed computing (cf.\
\cite{Peleg2000}), the \emph{run-time} of an asynchronous SA algorithm is
measured in terms of \emph{time units} scaled to the maximum of the time it
takes to deliver any message and the time between any two consecutive
activations of a node.
Refer to \cite{EmekW2013} for a more detailed description of the SA model.

The crux of the SA model is that the number of states in $Q$ and the size of
the message alphabet $\Sigma$ are constants independent of the size (and any
parameter) of the graph $G$.
Moreover, node $v$ cannot distinguish between its ports and in general, its
degree may be larger than $|Q|$ (and $|\Sigma|$).

\paragraph{Weakening the Communication Assumptions.}
The model defined in the current paper is a restriction of the model of
\cite{EmekW2013}, where the algorithm designer could choose an additional
constant \emph{bounding parameter}
$b \in \Integers_{> 0}$,
providing the nodes with the capability to count the number of ports holding
message
$\sigma \in \Sigma$
up to $b$.
In the current paper, the bounding parameter is set to
$b = 1$.
This model choice can be viewed as an asynchronous multi-frequency variant of
the \emph{beeping} communication model \cite{CornejoK2010, AfekABCHK2011}.

Moreover, in contrast to the existing SA literature, the communication graph
$G = (V, E)$
assumed in the current paper may include \emph{self-loops} of the form
$(v, v) \in E$
which means, in accordance with the definition of the SA model, that node $v$
admits port $\psi_{v}(v)$ that holds the last message received from itself.
Using the terminology of the beeping model literature (see, e.g.,
\cite{AfekABCHK2011}), the assumption that the communication graph is free of
self-loops corresponds to a \emph{sender collision detection}, whereas lifting
this assumption means that node $v$ may not necessarily distinguish its own
transmitted message from those of its neighbors.

It turns out that self-loops have a significant effect on the power of SA
algorithms.
Indeed, while a SA algorithm that solves the \emph{maximal independent set
(MIS)} problem with probability $1$ is presented in \cite{EmekW2013} under the
assumption that the graph is free of self-loops, we prove in
\Sect{}~\ref{sec:negative-results} that if the graph is augmented with
self-loops, then no SA algorithm can solve this problem with a bounded failure
probability.
To distinguish between the original model of \cite{EmekW2013} and the one
considered in the current paper, we hereafter denote the latter by
\SAloops{}.

\subsection{Results}
\label{sec:results}
Throughout, the number of nodes and the diameter of the graph $G$ are denoted
by $n$ and $D$, respectively.
We say that an event occurs \emph{with high probability (whp)} if its
probability is at least
$1 - n^{-c}$
for an arbitrarily large constant $c$.
Our main technical contribution is cast in the following two theorems.

\begin{theorem} \label{thm:leader-selection-positive}
For any constant $k$, there exists a \SAloops{} algorithm that solves the
$k$-leader selection problem in
$\tilde{O} (D)$
time whp.%
\footnote{The asymptotic notation
$\tilde{O} (\cdot)$
may hide
$\log^{O (1)} n$
factors.}
\end{theorem}

\begin{theorem} \label{thm:leader-selection-negative}
If the upper bound $k$ on the number of candidates may grow as a function of
$n$, then there does not exist a SA algorithm (operating on graphs with no
self-loops) that solves the $k$-leader selection problem with a failure
probability bounded away from $1$.
\end{theorem}

We emphasize that the failure probability of the \SAloops{} algorithm promised
in \Thm{}~\ref{thm:leader-selection-positive} (i.e., the probability that the
algorithm selects multiple leaders or that it runs for more than
$\tilde{O} (D)$
time) is inverse polynomial in $n$
even though each individual node does not (and cannot) possess any notion of
$n$ --- to a large extent, this, together with the termination detection
requirement, capture the main challenge in designing the promised
algorithm.\footnote{%
If we aim for a failure probability inverse polynomial in $k$ (rather than $n$)
and we do not insist on termination detection, then the problem is trivially
solved by the algorithm that simply assigns a random ID from a set of size
$k^{O (1)}$
to each candidate and then eliminates a candidate if it encounters an ID
larger than its own.
}
The theorem assumes that
$k = O (1)$
and hides the dependency of the algorithm's parameters on $k$.
A closer look at its proof reveals that our \SAloops{} algorithm uses local
memory and messages of size
$O (\log k)$
bits.
\Thm{}~\ref{thm:leader-selection-negative} asserts that the dependence of
these parameters on $k$ is unavoidable.
Whether this dependence can be improved beyond
$O (\log k)$
remains an open question.

\subsection{Additional Related Literature}
\label{sec:related-literature}
As mentioned earlier, the SA model was introduced by Emek and Wattenhofer
in \cite{EmekW2013} as an abstraction for distributed computing in networks of
devices whose computation and communication capabilities are far weaker than
those of a modern digital computer.
Their main focus was on distributed problems that can be solved in
sub-diameter (specifically,
$\log^{O (1)} n$)
time including MIS, tree coloring, coloring bounded degree graphs, and
maximal matching.
This remained the case also in \cite{EmekU2016}, where Emek and Uitto studied
SA algorithms for the MIS problem in dynamic graphs.
In contrast, the current paper considers the $k$-leader selection problem ---
an inherently global problem that requires
$\Omega (D)$
time.

Computational models based on networks of finite automata have been studied
for many years.
The best known such model is the extensively studied \emph{cellular automata}
that were introduced by Ulam and von Neumann \cite{vonNeumann1966} and became
popular with Martin Gardner's Scientific American column on Conway's
\emph{game of life} \cite{Gardner1970} (see also \cite{Wolfram2002}).

Another popular model that considers a network of finite automata is the
\emph{population protocols} model, introduced by Angluin et
al.~\cite{AngluinADFP2006} (see also
\cite{AspnesR2009, MichailCS2011}),
where the network entities communicate through a sequence of atomic pairwise
interactions controlled by a fair (adversarial or randomized) scheduler.
This model provides an elegant abstraction for networks of mobile devices with
proximity derived interactions and it also fits certain types of chemical
reaction networks \cite{Doty2014}.
Some work on population protocols augments the model with a graph defined over
the population's entities so that the pairwise interactions are restricted to
graph neighbors, thus enabling some network topology to come into play.
However, for the kinds of networks we are interested in, the fundamental
assumption of sequential atomic pairwise interactions may provide the
population protocol with unrealistic advantage over weaker message passing
variants (including the SA model) whose communication schemes do not enable a
node to interact with its individual neighbors independently.
Furthermore, population protocols are typically required to \emph{eventually
converge} to a correct output and are allowed to return arbitrary (wrong)
outputs beforehand, a significantly weaker requirement than the termination
detection requirement considered in this paper.

The neat \emph{amoebot model} introduced by Dolev et al.~\cite{DolevGRS2013}
also considers a network of finite automata in a (hexagonal) grid topology,
but in contrast to the models discussed so far, the particles in this network
are augmented with certain mobility capabilities, inspired by the amoeba
contraction-expansion movement mechanism.
Since its introduction, this model was successfully employed for the
theoretical investigation of self-organizing particle systems
\cite{SOPSWokshop2014, DerakhshandehGRSST2014, DerakhshandehGR2015,
DerakhshandehGSRRS2015, DerakhshandehGR2016, CannonDRR2016,
DaymudeDGPRSS2018},
especially in the context of \emph{programmable matter}.

Leader election is arguably the most fundamental problem in distributed systems
coordination and has been extensively studied from the early days of
distributed computing \cite{GallagerHS1983, FredericksonL1987}.
It is synonymous in most models to the construction of a spanning tree ---
another fundamental problem in distributed computing --- where the root is
typically the leader.
Leader election has many applications including
deadlock detection,
choosing a key/password distribution center,
and
implementing a distributed file system manager.
It also plays a key role in tasks requiring a reliable centralized
coordinating node, e.g., Paxos and Raft, where leader election is used for
consensus --- yet another fundamental distributed computing problem, strongly
related to leader election.
Notice that in our model, leader selection does not (and cannot) imply a
spanning tree, but it does imply consensus.

Angluin \cite{Angluin1980} proved that uniform algorithms cannot break
symmetry in a ring topology with success probability $1$.
Following this classic impossibility result, many symmetry breaking algorithms
(with and without termination detection) that relax some of the assumptions in
\cite{Angluin1980} were introduced
\cite{AbrahamsonAHK1986, AttiyaSW1988, ItaiR1990, SchieberS1994, AfekM1994}.
Itai and Rodeh \cite{ItaiR1990} were the first to design randomized leader
election algorithms with bounded failure probability in a ring topology,
assuming that the nodes know $n$.
Schieber and Snir \cite{SchieberS1994} and Afek and Matias \cite{AfekM1994}
extended their work to arbitrary topology graphs.

\section{\SAloops{} Algorithm for $k$-Leader Selection}
\label{sec:algorithm}
In this section, we present our \SAloops{} algorithm and establish
\Thm{}~\ref{thm:leader-selection-positive}.
We start with some preliminary definitions and assumptions presented in
\Sect{}~\ref{sec:preliminaries}.
\Sect{}\ \ref{sec:ball-growing} and \ref{sec:broadcast-echo} are dedicated to
the basic subroutines on which our algorithm relies.
The algorithm itself is presented in \Sect{}~\ref{sec:main-algorithm}, where
we also establish its correctness.
Finally, in \Sect{}~\ref{sec:run-time}, we analyze the algorithm's run-time.

\subsection{Preliminaries}
\label{sec:preliminaries}
As explained in \Sect{}~\ref{sec:model}, the execution in the SA (and
\SAloops{}) model is controlled by an asynchronous scheduler.
One of the contributions of \cite{EmekW2013} is a SA \emph{synchronizer}
implementation (cf.\ the $\alpha$-synchronizer of Awerbuch
\cite{Awerbuch1985}).
Given a synchronous SA algorithm $\mathcal{A}$ whose execution progresses in
fully synchronized \emph{rounds}
$t \in \Integers_{> 0}$
(with simultaneous wake-up), the synchronizer generates a valid (asynchronous)
SA algorithm $\mathcal{A}'$ whose execution progresses in \emph{pulses} such
that the actions taken by $\mathcal{A}'$ in pulse $t$ are identical to those
taken by $\mathcal{A}$ in round $t$.\footnote{%
We emphasize the role of the assumption that when the execution begins, the
ports hold the designated initial message $\sigma_{0}$.
Based on this assumption, a node can ``sense'' that some of its neighbors have
not been activated yet, hence synchronization can be maintained right from the
beginning.}
The synchronizer is designed so that the asynchronous algorithm $\mathcal{A}'$
has the same bounding parameter $b$
($= 1$
in the current paper) and asymptotic run-time as the synchronous algorithm
$\mathcal{A}$.

Although the model considered by Emek and Wattenhofer \cite{EmekW2013} assumes
that the graph has no self-loops, it is straightforward to apply their
synchronizer to graphs that do include self-loops, hence it can work also in
our \SAloops{} model.
Consequently, in what follows, we restrict our attention to synchronous
\SAloops{} algorithms.
Specifically, we assume that the execution progresses in synchronous rounds
$t \in \Integers_{> 0}$,
where in round $t$, each node $v$ \\
(1) receives the messages sent by its neighbors in round
$t - 1$; \\
(2)
updates its state;
and \\
(3)
sends a message to its neighbors (same message to all neighbors).

Since we make no effort to optimize the size of the messages used by our
algorithm, we assume hereafter that the message alphabet $\Sigma$ is identical
to the state space $Q$ and that node $v$ simply sends its current state to its
neighbors at the end of every round.
Nevertheless, for clarity of the exposition, we sometimes describe the
algorithm in terms of sending designated messages, recalling that this simply
means that the states of the nodes encode these messages.

To avoid cumbersome presentation, our algorithm's description does not get
down to the resolution of the state space $Q$ and transition function $\tau$.
It is straightforward though to implement our algorithm as a randomized finite
automaton, adhering to the model presented in \Sect{}~\ref{sec:model}.
In this regard, at the risk of stating the obvious, we remind the reader that
if $k$ is a constant, then a finite automaton supports arithmetic operations
modulo
$O (k)$.

In the context of the $k$-leader selection problem, we use the verb
\emph{withdraw} when referring to a node that ceases to be a candidate.

\subsection{The Ball Growing Subroutine}
\label{sec:ball-growing}
We present a generic \emph{ball growing} subroutine in graph
$G = (V, E)$
with at most $k$ candidates.
The subroutine is initiated at (all) the candidates, not necessarily
simultaneously, through designated signals discussed later on.
During its execution, some candidates may withdraw;
in the context of this subroutine, we refer to the surviving candidates
as \emph{roots}.

The ball growing subroutine assigns a \emph{level} variable
$\lambda(v) \in \{ 0, 1, \dots, M - 1 \}$
to each node $v$,
where
$M = 2 k + 2$.
Path
$P = (v_{1}, \dots, v_{q})$
in $G$ is called \emph{incrementing} if
$\lambda(v_{j + 1}) = \lambda(v_{j}) + 1 \bmod M$
for every
$1 \leq j \leq q - 1$.
The set of nodes reachable from a root $r$ via an incrementing path is referred
to as the \emph{ball} of $r$, denoted by $B(r)$.
We design this subroutine so that the following lemma holds.

\begin{lemma} \label{lem:ball-growing}
Upon termination of the ball growing subroutine, \\
(1)
every incrementing path is a shortest path (between its endpoints) in $G$; \\
(2)
every root belongs to exactly one ball (its own); and \\
(3)
every non-root node belongs to at least one ball.
\end{lemma}

\begin{IntuitionSpotlight}
A natural attempt to design the ball growing subroutine is to grow a breadth
first search tree around candidate $r$, layer by layer, so that node $v$ at
distance $d$ from $r$ is assigned with level variable
$\lambda(v) = d \bmod M$.
This is not necessarily possible though when multiple candidates exist:
What happens if the ball growing processes of different candidates reach $v$
in the same round?
What happens if these ball growing processes reach several adjacent nodes in
the same round?
If we are not careful, these scenarios may lead to incrementing paths that are
not shortest paths and even to cyclic incrementing paths.
Things become even more challenging considering the weak communication
capabilities of the nodes that may prevent them from distinguishing between
the ball growing processes of different candidates.
\end{IntuitionSpotlight}

The ball growing subroutine is implemented under the \SAloops{} model by
disseminating $\GrowBall(\ell)$ messages,
$\ell \in \{ 0, 1, \dots, M - 1 \}$,
throughout the graph.
Consider a candidate $r$ and let $s(r)$ be the round in which it is signaled to
invoke the ball growing subroutine.
If $r$ receives a $\GrowBall(\cdot)$ message in some round
$t \leq s(r)$,
then $r$ withdraws and subsequently follows the protocol like any other
non-root node;
otherwise, $r$ becomes a root in round $s(r)$.
If $s(r)$ is even (resp., odd), then $r$
assigns
$\lambda(r) \leftarrow 0$
(resp.,
$\lambda(r) \leftarrow 1$)
and sends a
$\GrowBall(\lambda(r))$
message.

Consider a non-root node $v$ and let $g(v)$ be the first round in which it
receives a $\GrowBall(\cdot)$ message.
Notice that $v$ may receive several $\GrowBall(\ell)$ messages with different
arguments $\ell$ in round $g(v)$ --- let $L$ be the set of all such arguments
$\ell$.
Node $v$ assigns
$\lambda(v) \leftarrow \ell'$
and sends a $\GrowBall(\ell')$ message at the end of round $g(v)$, where
$\ell'$ is chosen to be any integer in
$\{ 0, 1, \dots, M - 1 \}$
that satisfies: \\
(i)
$\ell' - 1 \bmod M \in L$;
and \\
(ii)
$\ell' + 1 \bmod M \notin L$. \\
This completes the description of the ball growing subroutine.
Refer to \Fig{}~\ref{fig:three-balls} for an illustration.

\begin{figure}
\begin{center}
\includegraphics[width=\textwidth]{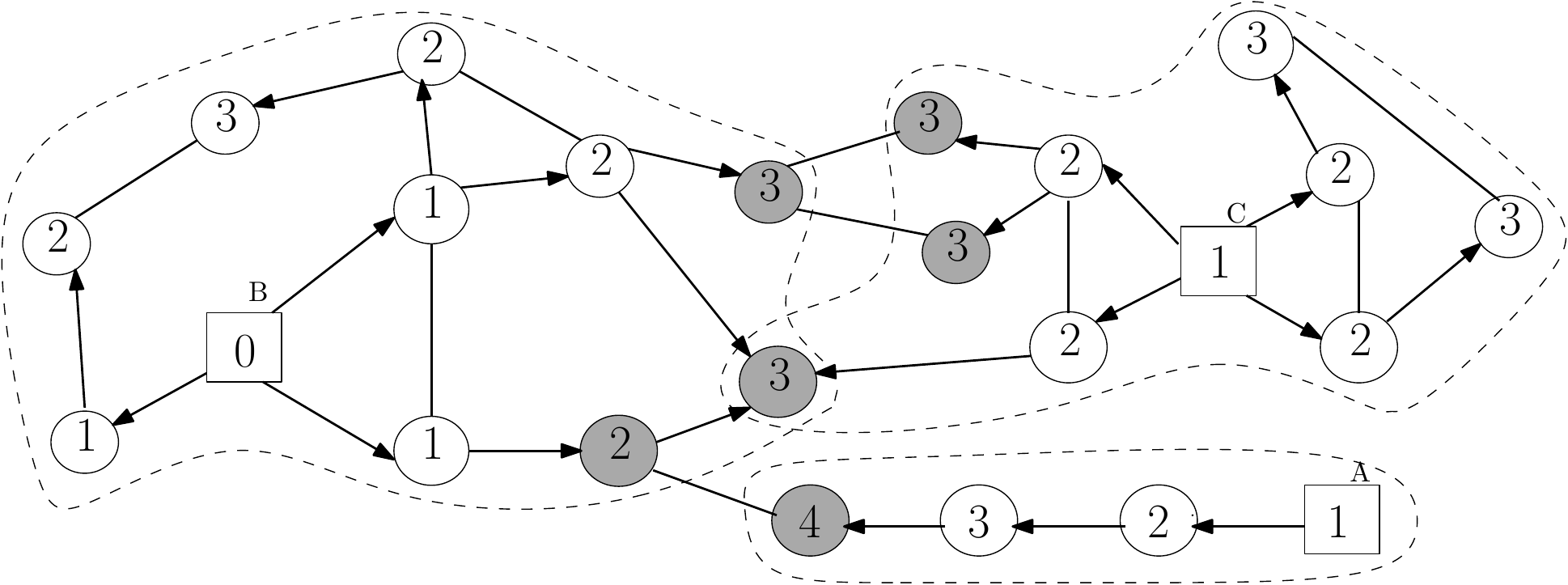}
\end{center}
\caption{\label{fig:three-balls}%
The result of a ball growing process invoked at
candidate A in round $1$,
candidate B in round $2$, and
candidate C in round $3$.
The level variables $\lambda(\cdot)$ are depicted by the numbers written
inside the nodes and the balls are depicted by the dashed curves.
The boundary nodes appear with a gray background.
The DAG $\vec{G}$ is depicted by the oriented edges.
}
\end{figure}

\begin{IntuitionSpotlight}
Condition (i) ensures that $v$ joins the ball $B(r)$ of some root $r$.
By condition (ii), nodes do not join $B(r)$ ``indirectly'' (this could have
led to incrementing paths that are not shortest paths).
\end{IntuitionSpotlight}

\begin{proof}[Proof of \Lem{}~\ref{lem:ball-growing}]
Consider a (root or non-root) node
$v \in V$
and let $p(v)$ be the round in which $v$ starts its active participation in
the ball growing process.
More formally, if $v$ is a root (i.e., it is a candidate signaled to invoke
the ball growing subroutine strictly before receiving any $\GrowBall(\cdot)$
message), then
$p(v) = s(v)$;
otherwise,
$p(v) = g(v)$.
The following properties are established by (simultaneous) induction on the
rounds:
\begin{DenseItemize}

\item
In any round
$t \geq p(v)$,
variable $\lambda(v)$ is even if and only if $p(v)$ is even.

\item
In any round
$t \geq p(v)$,
node $v$ has a neighbor $u$ with
$\lambda(u) = \lambda(v) - 1 \bmod M$
if and only if $v$ is not a root.

\item
In any round
$t \geq p(v)$,
node $v$ belongs to ball $B(r)$ for some root $r$.

\item
In any round
$t \geq p(v)$,
if
$v \in B(r)$
for some root $r$, then the incrementing path(s) that realize this relation
are shortest paths in the graph.

\item
If
$u, v \in B(r)$
for some root $r$ and
$p(u) = p(v)$,
then
$\lambda(u) = \lambda(v)$.

\item
The total number of different arguments $\ell$ in the
$\GrowBall(\ell)$ messages sent during a single round is at most $k$.

\item
Non-root node $v$ finds a valid value to assign to $\lambda(v)$ in round
$g(v) = p(v)$.

\end{DenseItemize}
The assertion follows.
\end{proof}

\begin{observation} \label{obs:ball-growing-run-time}
If $t$ is the earliest round in which the ball growing process is initiated at
some candidate, then the process terminates by round
$t +  O (D)$.
\end{observation}

\paragraph{Boundary Nodes.}
We will see in \Sect{}~\ref{sec:main-algorithm} that our algorithm detects
candidate multiplicity by identifying the existence of multiple balls in the
graph.
The key notion in this regard is the following one (see
\Fig{}~\ref{fig:three-balls}):
Node $v$ is said to be a \emph{boundary} node if \\
(1)
$v \in B(r) \cap B(r')$
for roots
$r \neq r'$;
or \\
(2)
$v \in B(r)$
for some root $r$ and there exists a neighbor $v'$ of $v$ such that
$v' \notin B(r)$.

\begin{observation} \label{obs:boundary-nodes}
If the graph has multiple roots, then every ball includes at least one
boundary node.
\end{observation}

Node $v$ is said to be a \emph{locally observable boundary} node if it has a
neighbor $v'$ such that
$\lambda(v') \notin \{ \lambda(v) + \ell \bmod M \mid \ell = -1, 0, +1 \}$.
Notice that by \Lem{}~\ref{lem:ball-growing}, there cannot be a ball that
includes both $v$ and $v'$ since then, at least one of the incrementing paths
that realize these inclusions is not a shortest path.
Therefore, a locally observable boundary node is in particular a boundary
node.

\paragraph{The Directed Acyclic Graph $\vec{G}$.}
Given two adjacent nodes $u$ and $v$, we say that $v$ is a \emph{child} of
$u$ and that $u$ is a \emph{parent} of $v$ if
$\lambda(v) = \lambda(u) + 1 \bmod M$;
a childless node is referred to as a \emph{leaf}.
This induces an orientation on a subset $F$ of the edges, say, from parents to
their children (up the incrementing paths), thus introducing a directed graph
$\vec{G}$ whose edge set is an oriented version of $F$ (see
\Fig{}~\ref{fig:three-balls}).
\Lem{}~\ref{lem:ball-growing} guarantees that $\vec{G}$ is acyclic (so, it is
a directed acyclic graph, abbreviated \emph{DAG}) and that it spans all nodes
in $V$.
Moreover, the sources and sinks of $\vec{G}$ are exactly the roots and leafs
of the ball growing subroutine, respectively, and the source-to-sink distances
in $\vec{G}$ are upper-bounded by the diameter $D$ of $G$.

We emphasize that the in-degrees and out-degrees in $\vec{G}$ are unbounded.
Nevertheless, the simplifying assumption that the messages sent by the
nodes encode their local states, including the level variables
$\lambda(\cdot)$ (see \Sect{}~\ref{sec:preliminaries}), ensures that node $v$
can distinguish between messages received from its children, messages received
from its parents, and messages received from nodes that are neither children
nor parents of $v$.

\subsection{Broadcast and Echo over $\vec{G}$}
\label{sec:broadcast-echo}
The assignment of level variables $\lambda(\cdot)$ by the ball growing
subroutine and the child-parent relations these variables induce provide a
natural infrastructure for \emph{broadcast and echo (B\&E)} over the
aforementioned DAG $\vec{G}$ so that the broadcast (resp., echo) process
progresses up (resp., down) the incrementing paths.
These are implemented based on $\Broadcast$ and $\Echo$ messages as follows.

The broadcast subroutine is initiated at (all) the roots, not necessarily
simultaneously, through designated signals discussed later on and root $r$
becomes \emph{broadcast ready} upon receiving such a signal.
A non-root node $v$ becomes broadcast ready in the first round in which
it receives $\Broadcast$ messages from all its parents.
A (root or non-root) node $v$ that becomes broadcast ready in round
$t^{b}_{0} = t^{b}_{0}(v)$
keeps sending $\Broadcast$ messages throughout the round interval
$[t^{b}_{0}, t^{b}_{1})$,
where
$t^{b}_{1} = t^{b}_{1}(v)$
is defined to be the first round (strictly) after $t^{b}_{0}$ in which \\
(i) $v$ receives $\Broadcast$ messages from all its children; and \\
(ii) $v$ does not receive a $\Broadcast$ message from any of its
parents. \\
(Notice that conditions (i) and (ii) are satisfied vacuously for the leaves
and roots, respectively.)

The echo subroutine is implemented in a reversed manner:
It is initiated at (all) the leaves, not necessarily simultaneously, after
their role in the broadcast subroutine ends so that leaf $v$ becomes
\emph{echo ready} in round $t^{b}_{1}(v)$.
A non-leaf node $v$ becomes \emph{echo ready} in the first round in which
it receives $\Echo$ messages from all its children.
A (leaf or non-leaf) node $v$ that becomes echo ready in round
$t^{e}_{0} = t^{e}_{0}(v)$
keeps sending $\Echo$ messages throughout the round interval
$[t^{e}_{0}, t^{e}_{1})$,
where
$t^{e}_{1} = t^{e}_{1}(v)$
is defined to be the first round (strictly) after $t^{e}_{0}$ in which \\
(i) $v$ receives $\Echo$ messages from all its parents; and \\
(ii) $v$ does not receive an $\Echo$ message from any of its
children. \\
(Notice that conditions (i) and (ii) are satisfied vacuously for the roots and
leaves, respectively.)

\begin{lemma} \label{lem:broadcast-echo}
The following properties hold for every B\&E process:
\begin{DenseItemize}

\item
Rounds
$t^{b}_{0}(v)$, $t^{b}_{1}(v)$, $t^{e}_{0}(v)$, and $t^{e}_{1}(v)$
exist and
$t^{b}_{0}(v) < t^{b}_{1}(v) \leq t^{e}_{0}(v) < t^{e}_{1}(v)$
for every node $v$.

\item
If node $v$ is reachable from node
$u \neq v$
in DAG $\vec{G}$, then
$t^{b}_{i}(u) < t^{b}_{i}(v)$
and
$t^{e}_{i}(u) > t^{e}_{i}(v)$
for
$i \in \{ 0, 1 \}$.

\item
If $t$ is the latest round in which the process is initiated at some root,
then the process terminates by round
$t +  O (D)$.

\end{DenseItemize}
\end{lemma}
\begin{proof}
Follows since $\vec{G}$ is a DAG and all paths in $\vec{G}$ are shortest
paths.
\end{proof}

\paragraph{Auxiliary Conditions.}
In the aforementioned implementation of the broadcast (resp., echo) subroutine,
being broadcast (resp., echo) ready is both a necessary and sufficient
condition for a node to start sending $\Broadcast$ (resp.,
$\Echo$) messages.
In \Sect{}~\ref{sec:main-algorithm}, we describe variants of this subroutine
in which being broadcast (resp., echo) ready is a necessary, but not
necessarily sufficient, condition and the node starts sending
$\Broadcast$ (resp., $\Echo$) messages only after additional
conditions, referred to later on as \emph{auxiliary conditions}, are
satisfied.

\paragraph{Acknowledged Ball Growing.}
As presented in \Sect{}~\ref{sec:ball-growing}, the ball growing subroutine
propagates from the roots to the leaves.
To ensure that root $r$ is signaled when the construction of its ball
$B(r)$ has finished (cf.\ termination detection), $r$ initiates a B\&E process
one round after it invokes the ball growing subroutine.
The valid operation of this process is guaranteed since the ball growing
process propagates at least as fast as the B\&E process.
We call the combined subroutine \emph{acknowledged ball growing}.

\subsection{The Main Algorithm}
\label{sec:main-algorithm}
Our $k$-leader selection algorithm consists of two \emph{phases} executed
repeatedly in alternation:
\begin{DenseItemize}

\item
phase $0$, a.k.a.\ the \emph{detection} phase, that detects the
existence of multiple candidates whp;
and

\item
phase $1$, a.k.a.\ the \emph{elimination} phase, in which all candidates but
one withdraw with probability at least
$1 / 4$.

\end{DenseItemize}
Starting with a detection phase, the algorithm executes the phases in
alternation until the first detection phase that does not detect candidate
multiplicity.
Each node $v$ maintains a \emph{phase} variable
$\phi(v) \in \{ 0, 1 \}$
that indicates $v$'s current phase.

The two phases follow a similar structure:
The (surviving) candidates start by initiating an acknowledged ball growing
process.
Among its other ``duties'', this ball growing process is responsible for
updating the phase variables $\phi(\cdot)$ of the nodes:
node $v$ with
$\phi(v) = p$
that receives a $\GrowBall(\cdot)$ message from node $u$ with
$\phi(u) = p + 1 \bmod 2$
assigns
$\phi(v) \leftarrow p + 1 \bmod 2$.
When updating the phase variable $\phi(v)$ to
$\phi(v) = p + 1 \bmod 2$,
node $v$ ceases to participate in phase $p$, resetting all phase $p$
variables.
Recalling the definition of the ball growing subroutine (see
\Sect{}~\ref{sec:ball-growing}), this means in particular that if a candidate
$r$ with
$\phi(r) = p$
receives a $\GrowBall(\cdot)$ message from node $u$ with
$\phi(u) = p + 1 \bmod 2$,
then $r$ withdraws and subsequently follows the protocol like any other
non-root node.

\begin{IntuitionSpotlight}
The ball growing process of phase
$p + 1 \bmod 2$
essentially ``takes control'' over the graph and ``forcibly'' terminates
phase $p$ (at nodes where it did not terminate already).
We design the algorithm to ensure that at any point in time, there is at most
one $p$ value for which there is an ongoing ball growing process in the graph
(otherwise, we may get to undesired situations such as all candidates
withdrawing).
\end{IntuitionSpotlight}

Upon termination of the acknowledged ball growing process, the roots run $2 k$
back-to-back \emph{B\&E iterations}, initiating the broadcast process of the
next B\&E iteration one round after the echo process of the previous B\&E
iteration terminates
(the choice of the parameter
$2 k$
will become clear soon).
Each node $v$ maintains a variable
$\iota(v) \in \{ 0, 1, \dots, 2 k \}$
that stores $v$'s current B\&E iteration.
This variable is initialized to
$\iota(v) \leftarrow 0$
during the acknowledged ball growing process (considered hereafter as B\&E
iteration $0$) and incremented subsequently from
$i - 1$
to $i$ when $v$ becomes broadcast ready in B\&E iteration $i$ (see
\Sect{}~\ref{sec:broadcast-echo}).
A phase ends when the echo process of B\&E iteration
$2 k$
terminates.

The $\iota(\cdot)$ variables may differ across the graph and to keep the
B\&E iterations in synchrony, we augment the B\&E subroutines with the
following auxiliary conditions (see \Sect{}~\ref{sec:broadcast-echo}):
Node $v$ with
$\iota(v) = i$
(i.e., in B\&E iteration $i$) does not start to send $\Broadcast$ (resp.,
$\Echo$) messages as long as it has a non-child (resp., non-parent) neighbor
$u$ with
$\iota(u) = i - 1$.%
\footnote{This can be viewed as imposing the $\alpha$-synchronizer of
\cite{Awerbuch1985} on the B\&E iterations of the balls.}
We emphasize that this includes neighbors $u$ that are neither children nor
parents of $v$.

For the sake of the next observation, we globally map the B\&E iterations to
\emph{sequence numbers} so that B\&E iterations
$0, 1, \dots, 2 k$
of the first phase (which is a detection phase) are mapped to sequence
numbers
$1, 2, \dots, 2 k + 1$,
respectively,
B\&E iterations
$0, 1, \dots, 2 k$
of the second phase (which is an elimination phase) are mapped to sequence
numbers
$2 k + 2, 2 k + 3, \dots, 4 k + 2$,
respectively,
and so on.
Let $\sigma(v)$ be a variable (defined only for the sake of the analysis)
indicating the sequence number of node $v$'s current B\&E iteration.

\begin{observation} \label{obs:iteration-synchrony-roots}
For every two roots $r$ and $r'$, we have
$|\sigma(r) - \sigma(r')| \leq k - 1$.
\end{observation}

We say that round $t$ is \emph{$0$-dirty} (resp., \emph{$1$-dirty}) if some
node $v$ with
$\phi(v) = 0$
(resp., $\phi(v) = 1$)
sends a $\GrowBall(\cdot)$ message in round $t$;
the round is said to be \emph{clean} if it is neither $0$-dirty nor $1$-dirty.
\Obs{}~\ref{obs:iteration-synchrony-roots} implies that if
$\phi(r) = p$
and
$\iota(r) = k$
for some root $r$ in round $t$,
then
$\phi(r') = p$
and
$1 \leq \iota(r') \leq 2 k - 1$
for any other root $r'$ in round $t$, hence the ball growing process of this
phase has already ended and the ball growing process of the next phase has
not yet started.

\begin{corollary} \label{cor:clean-round}
Let $t_{0}$ and $t_{1}$ be some $0$-dirty and $1$-dirty rounds, respectively.
If
$t_{0} \leq t_{1}$
(resp.,
$t_{1} \leq t_{0}$),
then there exists some
$t_{0} < t' < t_{1}$
(resp.,
$t_{1} < t' < t_{0}$)
such that round $t'$ is clean.
\end{corollary}

\subsubsection{The Detection Phase}
\label{sec:detection-phase}
In the detection phase, the nodes test for candidate multiplicity in the
graph.
If the graph contains a single candidate $r$, then the algorithm terminates
upon completion of this phase and $r$ is declared to be the leader.
Otherwise, certain boundary nodes (see \Sect{}~\ref{sec:ball-growing}) realize
whp that multiple balls exist in their neighborhoods and signal the roots that
they should proceed to the elimination phase (rather than terminate the
algorithm) upon completion of the current detection phase.
This signal is carried by $\Proceed$ messages delivered from the boundary
nodes to the roots of their balls down the incrementing paths in conjunction
with the $\Echo$ messages of the (subsequent) B\&E iterations.

For the actual candidate multiplicity test, once all nodes in the (inclusive)
neighborhood of node $v$ participate in the detection
phase, node $v$ checks if it is a locally observable boundary node and
triggers a $\Proceed$ message delivery if it is.
As the name implies, this check can be performed (locally) under the
\SAloops{} model assuming that the messages sent by the nodes encode their
local states, including the level variables.

\begin{IntuitionSpotlight}
Although every locally observable boundary node is a boundary node, not all
boundary nodes are locally observable:
a node may belong to several different balls or two adjacent nodes with the
same level variable may belong to different balls.
For this kind of scenarios, randomness is utilized to break symmetry between
the candidates and identify (some of) the boundary nodes.
\end{IntuitionSpotlight}

Consider some root $r$ with
$\phi(r) = 0$
upon termination of the acknowledged ball growing subroutine and recall that
at this stage, $r$ runs
$2 k$
back-to-back B\&E iterations.
In each round of these
$2 k$
B\&E iterations, $r$ picks some \emph{symbol} $s$ uniformly at random (and
independently of all other random choices) from a sufficiently large (yet
constant size) symbol space $\mathcal{S}$ and sends a $\RandSymbol(s)$
message.
This can be viewed as a random symbol \emph{stream}
$S_{r} \in \mathcal{S}^{*}$
that $r$ generates, round by round, and sends to its children.

The random symbol streams $S_{r}$ are disseminated throughout $B(r)$ and
utilized by the nodes (the boundary nodes in particular) to test for candidate
multiplicity.
For clarity of the exposition, it is convenient to think of a node $v$ that
does not send a $\RandSymbol(s)$ message,
$s \in \mathcal{S}$,
as if it sends a $\RandSymbol(\bot)$ message for the default symbol
$\bot \notin \mathcal{S}$.
The mechanism in charge of disseminating $S_{r}$ up the incrementing paths
works as follows:
If non-root node $v$ with
$\phi(v) = 0$
receives $\RandSymbol(s)$ messages with the same argument $s$ from all its
parents at the beginning of round $t$, then $v$ sends a $\RandSymbol(s)$
message at the end of round $t$;
in all other cases, $v$ sends a $\RandSymbol(\bot)$ message.

Throughout this process, each node $v$ verifies that \\
(1)
all $\RandSymbol(s)$ messages sent by $v$'s parents in round $t$ carry the
same argument $s$;
and \\
(2)
any $\RandSymbol(s)$ message sent by a neighbor $u$ of $v$ with
$\lambda(u) = \lambda(v)$
in round $t$ carries the same argument $s$ as in the $\RandSymbol(s)$ message
that $v$ sends in round $t$ (this is checked by $v$ in round
$t + 1$). \\
If any of these two conditions does not hold, then $v$ triggers a $\Proceed$
message delivery.
A root that completes all
$2 k$
B\&E iterations in the detection phase without receiving any $\Proceed$
message terminates the algorithm and declares itself as the leader.

\begin{IntuitionSpotlight}
Since the aforementioned random tests should detect candidate multiplicity whp
(i.e., with error probability inverse polynomial in $n$) and since the size
of the symbol space $\mathcal{S}$ from which the random symbol streams $S_{r}$
are generated is bounded, it follows that the length of the random symbol
streams must be
$|S_{r}| \geq \Omega (\log n)$.
How can we ensure that
$|S_{r}| \geq \Omega (\log n)$
if the nodes cannot count beyond some constant?
\end{IntuitionSpotlight}

To ensure that the random symbol stream $S_{r}$ is sufficiently long, we
augment the echo subroutine invoked during B\&E iteration $k$ of the detection
phase (out of the
$2 k$
B\&E iterations in this phase) with one additional auxiliary condition
referred to as the \emph{geometric auxiliary condition}:
Consider some node $v$ with
$\phi(v) = 0$
and
$\iota(v) = k$
(i.e., in the $k$-th B\&E iteration of the detection phase)
and suppose that it becomes echo ready (for B\&E iteration $k$) in round
$t_{0}$.
Then, $v$ tosses a fair coin
$c(t) \in_{r} \{ 0, 1 \}$
in each round
$t \geq t_{0}$
until the first round $t'$ for which
$c(t') = 1$;
node $v$ does not send $\Echo$ messages until round $t'$.
This completes the description of the detection phase.

\begin{lemma} \label{lem:detection-phase-safety}
If multiple roots start a detection phase, then all of them receive a
$\Proceed$ message before completing their (respective)
$2 k$
B\&E iterations whp.
\end{lemma}

\begin{IntuitionSpotlight}
The proof's outline is as follows.
We use the geometric auxiliary conditions to argue that there exists some root
that spends
$\Omega (\log n)$
rounds in B\&E iteration $k$ whp.
Employing \Obs{}~\ref{obs:iteration-synchrony-roots}, we conclude that
the random symbol stream generated by every root $r$ is
$\Omega (\log n)$-long
whp.
Conditioned on that, we prove that there exists some boundary node
$v \in B(r)$
that triggers a $\Proceed$ message delivery whp and that the corresponding
$\Proceed$ message is delivered to $r$ before the phase ends.
\end{IntuitionSpotlight}

\begin{proof}[Proof of \Lem{}~\ref{lem:detection-phase-safety}]
Fix some detection phase.
For a root $r$, let $c_{r}$ be the number of rounds $r$ spends in B\&E
iterations
$1, 2, \dots, 2 k - 1$,
that is, the number of rounds in which
$1 \leq \iota(r) \leq 2 k - 1$
(during this detection phase).
We first argue that
$c_{r} \geq \Omega (\log n)$
for all roots $r$ whp.
To that end, let $X_{v}$ be the number of rounds in which node $v$ is
prevented from sending its $\Echo$ messages in B\&E iteration $k$ due to the
geometric auxiliary condition
($t' - t_{0}$
in the aforementioned notation of the geometric auxiliary condition) and
notice that this auxiliary condition is designed so that $X_{v}$ is a
geometric random variable with parameter
$1 / 2$.
Therefore,
\[
\Pr \left( \bigwedge_{v \in V} X_{v} < \log (n) / 2 \right)
\, = \,
\left( 1 - 2^{-\log (n) / 2} \right)^{n}
\, = \,
\left( 1 - 1 / \sqrt{n} \right)^{n}
\, \leq \,
e^{-\sqrt{n}} \, .
\]

Condition hereafter on the event that
$X_{v^{*}} \geq \log (n) / 2$
for some node $v^{*}$, namely, $v^{*}$ is prevented from sending its $\Echo$
messages (in B\&E iteration $k$) for at least
$\log (n) / 2 = \Omega (\log n)$
rounds.
Let $r^{*}$ be a root such that
$v^{*} \in B(r^{*})$.
By the definition of auxiliary conditions, B\&E iteration $k$ of $r^{*}$ takes
at least
$\Omega (\log n)$
rounds.
\Obs{}~\ref{obs:iteration-synchrony-roots} guarantees that by the time
$r^{*}$ starts B\&E iteration $k$, every other root must have already started
B\&E iteration $1$ (of this detection phase).
Moreover, no root can start B\&E iteration
$2 k$
before $r^{*}$ finishes B\&E iteration $k$.
We conclude that every root $r$ spends at least
$\Omega (\log n)$
rounds in B\&E iterations
$1, 2, \dots, 2 k - 1$,
thus establishing the argument.

Let $Z_{r}$ be the prefix of the random symbol stream $S_{r}$ generated by
root $r$ during the first
$c_{r} - 1$
rounds it spends in B\&E iterations
$1, 2, \dots, 2 k - 1$,
i.e., during all but the last round of these B\&E iterations (the reason for
this missing round is explained soon), and let
$z_{r} = |Z_{r}|$.
We have just showed that
$z_{r} = c_{r} - 1 \geq \Omega (\log n)$
for all roots $r$ whp.

The assertion is established by proving that if multiple roots $r$ exist in
the graph and
$z_{r} \geq \Omega (\log n)$
for all of them, then for every root $r$, there exists some node
$v \in B(r)$
that triggers a $\Proceed$ message delivery while
$\iota(v) \leq 2 k - 1$
whp.
Indeed, if the $\Proceed$ message delivery is triggered by $v$ while
$\iota(v) \leq 2 k - 1$,
then a $\Proceed$ message is delivered to $r$ with the $\Echo$ messages of
B\&E iteration
$2 k$
at the latest, thus $r$ does not terminate the algorithm at the end of this
detection phase and by the union bound, this holds simultaneously for all
roots $r$ whp.

To that end, recall that node $v$ sends a $\RandSymbol(s)$ message with some
symbol
$s \in \mathcal{S} \cup \{ \bot \}$
in every round of the detection phase.
In the scope of this proof, we say that $v$ \emph{posts} the symbol stream
$(s_{1}, \dots, s_{z})$
in rounds
$t_{1}, \dots, t_{z}$
if $s_{j}$ is the argument of the $\RandSymbol(\cdot)$ message sent by $v$ in
round $t_{j}$ for every
$1 \leq j \leq z$.

Consider some root $r$ and let $v$ be a boundary node in $B(r)$ that minimizes
the distance to $r$.
If $v$ is locally observable, then it triggers a $\Proceed$ message delivery
(deterministically) already when
$\iota(v) = 0$,
so assume hereafter that $v$ is not locally observable.
Let $Q$ be an incrementing
$(r, v)$-path
and denote the length of $Q$ by $q$.
Taking $\hat{t}$ to be the round in which B\&E iteration $1$ of $r$ begins,
recall that $r$ posts $Z_{r}$ in rounds
$\hat{t}, \hat{t} + 1, \dots, \hat{t} + z_{r} - 1$.
The choice of $v$ ensures that all nodes of $Q$ other than $v$ are not
boundary nodes, therefore if
$q \geq 1$
(i.e., if
$v \neq r$),
then the node that precede $v$ along $Q$ --- denote it by $u$ --- posts
$Z_{r}$ in rounds
$\hat{t} + q - 1, \hat{t} + q, \dots, \hat{t} + q + z_{r} - 2$.
Moreover, by the definition of $Z_{r}$, specifically, by the choice of
$z_{r} = c_{r} - 1$,
we know that
$0 \leq \iota(v) \leq 2 k - 1$
(and
$\phi(v) = 0$)
in all rounds
$\hat{t} \leq t \leq \hat{t} + q + z_{r}$.

If $v$ belongs to multiple balls, which necessarily means that
$v \neq r$
and
$q \geq 1$
(see \Lem{}~\ref{lem:ball-growing}),
then $v$ has another parent
$u' \neq u$
such that
$u' \in B(r')$
for some root
$r' \neq r$.
The probability that $u'$ posts $Z_{r}$ in rounds
$\hat{t} + q - 1, \hat{t} + q, \dots, \hat{t} + q + z_{r} - 2$
is at most
$|\mathcal{S}|^{-z_{r}}$.
Otherwise, if $v$ belongs only to ball $B(r)$, then all its parents post
$Z_{r}$ in rounds
$\hat{t} + q - 1, \hat{t} + q, \dots, \hat{t} + q + z_{r} - 2$
(this holds vacuously if
$q = 0$
and
$v = r$
has no parents),
thus $v$ posts $Z_{r}$ in rounds
$\hat{t} + q, \hat{t} + q + 1, \dots, \hat{t} + q + z_{r} - 1$.
Since $v$ is a non-locally observable boundary node (that belongs exclusively
to ball $B(r)$), it must have a neighbor $v'$ with
$\lambda(v') = \lambda(v)$
such that
$v' \notin B(r)$.
The probability that $v'$ posts $Z_{r}$ in rounds
$\hat{t} + q, \hat{t} + q + 1, \dots, \hat{t} + q + z_{r} - 1$
is at most
$|\mathcal{S}|^{-z_{r}}$
as well.
Therefore, the probability that $v$ does not trigger a $\Proceed$ message
delivery while
$\iota(v) \leq 2k - 1$
is upper-bounded by
$|\mathcal{S}|^{-z_{r}}$
which completes the proof since
$z_{r} \geq \Omega (\log n)$
and since $|\mathcal{S}|$ is an arbitrarily large constant.
\end{proof}

\subsubsection{The Elimination Phase}
\label{sec:elimination-phase}
In the elimination phase, each candidate $r$ picks a \emph{priority} $\pi(r)$
uniformly at random (and independently) from a totally ordered priority space
$\mathcal{P}$;
a candidate whose priority is (strictly) smaller than
$\pi_{\max} = \max_{r} \pi(r)$
is withdrawn.
Taking the priority space to be
$\mathcal{P} = \{ 1, \dots, k\}$,
it follows by standard balls-in-bins arguments that the probability that
exactly one candidate picks priority $k$, which implies that exactly one
candidate survives, is at least
$1 / 4$
(in fact, it tends to
$1 / 4$
as
$k \rightarrow \infty$).

\begin{IntuitionSpotlight}
The priorities of the candidates are disseminated in the graph so that
candidate $r$ withdraws if it encounters a priority
$\pi > \pi(r)$.
This is implemented on top of the ball growing subroutine invoked at
the beginning of the elimination phase so that the ball growing process of
root $r$ ``consumes'' the ball of root $r'$ if
$\pi(r) > \pi(r')$,
eventually reaching $r'$ and instructing it to withdraw.
The structure of the phase (specifically, the
$2 k$
B\&E iterations that follow the ball growing process) guarantees
that only roots $r$ with
$\pi(r) = \pi_{\max}$
reach the end of the phase (without being withdrawn).
\end{IntuitionSpotlight}

We augment the ball growing subroutine invoked at the beginning of the
elimination phase with the following mechanism:
When candidate $r$ is signaled to invoke the ball growing subroutine (so that
it becomes a root), it appends its priority $\pi(r)$ to the $\GrowBall(\cdot)$
message it sends.
A non-root node $v$ that joins the ball of $r$ records $r$'s priority in
variable
$\pi(v) \leftarrow \pi(r)$.
A (root or non-root) node $v$ that receives a $\GrowBall(\cdot)$ message with
priority (strictly) larger than $\pi(v)$, behaves as if this is the first
$\GrowBall(\cdot)$ message it receives in this phase.
In particular, $v$ resets all the variables of this phase and (re-)joins a ball
from scratch.
If $v$ is a root, then it also withdraws.

Notice that \Obs{}~\ref{obs:iteration-synchrony-roots} still holds for
the aforementioned augmented implementation of the ball growing subroutine.
Therefore, when root $r$ reaches B\&E iteration $k$, i.e.,
$\iota(r) = k$,
all other roots $r'$ are in some B\&E iteration
$1 \leq \iota(r') \leq 2 k - 1$
which means that there is no ``active'' ball growing processes in the graph,
that is, the current round is clean (of $\GrowBall(\cdot)$ messages).
Since a candidate $r$ with
$\pi(r) < \pi_{\max}$
is certain to be withdrawn by some $\GrowBall(\cdot)$ message appended with
priority
$\pi > \pi(r)$, we obtain the following observation.

\begin{observation} \label{obs:elimination-phase}
If root $r$ completes its
$2 k$
B\&E iterations in an elimination phase, then with probability
at least
$1 / 4$,
no other candidates exist in the graph.
\end{observation}

\subsection{Run-Time}
\label{sec:run-time}
The correctness of our algorithm follows from
\Lem{}~\ref{lem:detection-phase-safety} and
\Obs{}~\ref{obs:elimination-phase}.
To establish \Thm{}~\ref{thm:leader-selection-positive}, it remains to
analyze the algorithm's run-time.

The first thing to notice in this regard is that the geometric auxiliary
condition does not slow down the $k$-th iteration of the detection phase by
more than an
$O (\log n)$
factor whp.
Combining \Obs{}~\ref{obs:ball-growing-run-time} with
\Lem{}~\ref{lem:broadcast-echo}, we can prove by induction on the phases that
the $j$-th phase (for
$j \leq n^{O (1)}$)
ends by round
$O (D (k + \log n))$
whp, which is
$O (D \log n)$
assuming that $k$ is fixed.
The analysis is completed due to \Obs{}~\ref{obs:elimination-phase}
ensuring that the algorithm terminates after
$O (\log n)$
elimination phases whp.

\section{Negative Results}
\label{sec:negative-results}
We now turn to establish some negative results that demonstrate the necessity
of the assumption that
$k = O (1)$.
Our attention in this section is restricted to SA and \SAloops{} algorithms
operating under a fully synchronous scheduler on graph families
$\{ L_{n} \}_{n \geq 1}$
and
$\{ L^{\circlearrowleft}_{n} \}_{n \geq 1}$,
where $L_{n}$ is a simple path of $n$ nodes and $L^{\circlearrowleft}_{n}$ is
$L_{n}$ augmented with self-loops.

The main lemma established in this section considers the \emph{$k$-candidate
binary consensus} problem, a version of the classic binary consensus problem
\cite{FischerLP1985}.
In this problem, each node $v$ gets a binary input
$\mathrm{in}(v) \in \{ 0, 1 \}$
and returns a binary output
$\mathrm{out}(v) \in \{ 0, 1 \}$
under the following two constraints:
(1)
all nodes return the same output;
and
(2)
if the nodes return output
$b \in \{ 0, 1 \}$,
then there exists some node $v$ such that
$\mathrm{in}(v) = b$.
In addition, at most $k$ (and at least $1$) nodes are initially marked as
candidates (thus distinguished from the rest of the nodes).
We emphasize that the marked candidates do not affect the validity of the
output.
Since a $k$-leader selection algorithm clearly implies a $k$-candidate binary
consensus algorithm, Theorem~\ref{thm:leader-selection-negative} is
established by proving Lemma~\ref{lem:negative-binary-consensus}.
Note that the proof of this lemma is based on a probabilistic
indistinguishability argument, similar to those used in many distributed
computing negative results, starting with the classic result of Itai and Rodeh
\cite{ItaiR1990}.

\begin{lemma} \label{lem:negative-binary-consensus}
If the upper bound $k$ on the number of candidates may grow as a function of
$n$, then there does not exist a SA algorithm that solves the
$k$-candidate binary consensus problem on the graphs in
$\{ L_{n} \}_{n \geq 1}$
with a failure probability bounded away from $1$.
\end{lemma}
\begin{proof}
Assume by contradiction that there exists such an algorithm $\mathcal{A}$ and
let $\Sigma$ denote its message alphabet.
For
$b = 0, 1$,
consider the execution of $\mathcal{A}$ on an instance that consists of path
$L_{2}$, where node $v_{1}$ is a candidate, node $v_{2}$ is not a candidate, and
$\mathrm{in}(v_{1}) = \mathrm{in}(v_{2}) = b$.
By definition, there exist constants
$p_{b} > 0$
and
$\ell_{b}$
and
message sequences
$S_{b, 1}, S_{b, 2} \in \Sigma^{\ell_{b}}$
such that when $\mathcal{A}$ runs on this instance, with probability at least
$p_{b}$, node $v_{j}$,
$j \in \{ 1, 2 \},$
reads message $S_{b, j}(t)$ in its (single) port in round
$t = 1, \dots, \ell_{b}$
and outputs
$\mathrm{out}(v_{j}) = b$
at the end of round $\ell_{b}$.

Now, consider graph $L_{n}$ for some sufficiently large $n$ (whose value will
be determined later on) and consider a subgraph of $L_{n}$, referred to as a
\emph{$Q_{b}$-gadget}, that consists of
$2 \ell_{b} + 2$
contiguous nodes
$v_{1}, \dots, v_{2 \ell_{b} + 2}$
of the underlying path $L_{n}$, all of which receive input
$\mathrm{in}(v_{i}) = b$.
Moreover, the nodes
$v_{1}, \dots, v_{2 \ell_{b} + 2}$
are marked as candidates in an alternating fashion so that if $v_{i}$ is a
candidate, then
$v_{i + 1}$
is not a candidate, constrained by the requirement that
$v_{\ell_{b} + 1}$
is a candidate (and
$v_{\ell_{b} + 2}$
is not).
The key observation is that when $\mathcal{A}$ runs on $L_{n}$, with
probability at least
$q_{b} = p_{b}^{2 \ell_{b} + 2}$,
the nodes
$v_{\ell_{b} + 1}$
and
$v_{\ell_{b} + 2}$
of the $Q_{b}$-gadget read messages $S_{b, 1}(t)$ and $S_{b, 2}(t)$,
respectively, in (all) their ports in round
$t = 1, \dots, \ell_{b}$
and output $b$ at the end of round $\ell_{b}$, independently of the random
bits of the nodes outside the $Q_{b}$-gadget.

Fix
$\ell = \ell_{0} + \ell_{1} + 2$
and define a \emph{$Q$-gadget} to be a subgraph of $L_{n}$ that consists of a
$Q_{0}$-gadget appended to a $Q_{1}$-gadget, so, in total, the $Q$-gadget is a
(sub)path that contains
$2 \ell_{0} + 2 \ell_{1} + 4 = 2 \ell$
nodes, $\ell$ of which are candidates.
Following the aforementioned observation, when $\mathcal{A}$ runs on $L_{n}$,
with probability at least
$q = q_{0} \cdot q_{1}$,
some nodes in the $Q$-gadget output $0$ and others output $1$;
we refer to this (clearly invalid) output as a \emph{failure} event of the
$Q$-gadget.

Since $p_{0}$, $p_{1}$, $\ell_{0}$, and $\ell_{1}$ are constants that depend
only on $\mathcal{A}$,
$\ell = \ell_{0} + \ell_{1} + 2$,
$q_{0} = p_{0}^{2 \ell_{0} + 2}$
and
$q_{1} = p_{1}^{2 \ell_{1} + 2}$
are also constants that depend only on $\mathcal{A}$, and thus
$q = q_{0} \cdot q_{1}$
is also a constant that depends only on $\mathcal{A}$.
Take $z$ to be an arbitrarily large constant.
If $n$ is sufficiently large, then we can embed
$y = \lceil z / q \rceil$
pairwise disjoint $Q$-gadgets in $L_{n}$.
Indeed, these $Q$-gadgets account to a total of
$\ell \cdot y$
candidates and recalling that $z$, $q$, and $\ell$ are constants, this number
is smaller than
$k = k(n)$
for sufficiently large $n$.
When $\mathcal{A}$ runs on $L_{n}$, each of these $y$ $Q$-gadgets fails with
probability at least $q$ (independently).
Therefore, the probability that all nodes return the same binary output is at
most
$(1 - q)^{y}$.
The assertion follows since this expression tends to $0$ as
$y \rightarrow \infty$
which is obtained as
$z \rightarrow \infty$.
\end{proof}

The proof of \Lem{}~\ref{lem:negative-binary-consensus} essentially shows
that no SA algorithm can distinguish between $L_{2}$ and $L_{n}$ with a
bounded failure probability.
When the path is augmented with self-loops, we can use a very similar line of
arguments to show that no \SAloops{} algorithm can distinguish between
$L^{\circlearrowleft}_{1}$ and $L^{\circlearrowleft}_{n}$ with a bounded
failure probability.
This allows us to establish the following lemma that should be contrasted with
the SA MIS algorithm of \cite{EmekW2013} that works on general topology graphs
(with no self-loops) and succeeds with probability $1$.

\begin{lemma} \label{lem:negative-mis}
There does not exist a \SAloops{} algorithm that solves the MIS problem on the
graphs in
$\{ L^{\circlearrowleft}_{n} \}_{n \geq 1}$
with a failure probability bounded away from $1$.
\end{lemma}

\clearpage

\bibliographystyle{alpha}
\bibliography{references}

\end{document}